\documentclass[runningheads]{llncs}
\usepackage[T1]{fontenc}
\usepackage{graphicx}

\usepackage{xcolor}
\usepackage{hyperref}
\usepackage{makecell}
\usepackage{comment}
\usepackage{enumitem}
\usepackage{booktabs}
\usepackage{amsmath}
\usepackage[commentsnumbered]{algorithm2e}
\DontPrintSemicolon

\SetCommentSty{commentFont}
\SetKwProg{algo}{}{:}{}
\usepackage{etoolbox}
\patchcmd{\SetProgSty}{ArgSty}{ProgSty}{}{}
\SetProgSty{upshape}

\newtheorem{invariant}{Invariant}

\newcommand{\Points}{\mathcal{P}}
\newcommand{\Tree}{\mathcal{T}}
\newcommand{\Hull}{\mathcal{H}}

\DeclareMathOperator*{\argmax}{arg\,max}

\begin{document}
	
\title{Online computation of normalized substring complexity}

\author{Gregory Kucherov\inst{1}\orcidID{0000-0001-5899-5424} \and
	Yakov Nekrich\inst{2}\orcidID{0000-0003-3771-5088} 
}

\institute{
	LIGM, CNRS, Universit\'e Gustave Eiffel, Marne-la-Vall\'ee, France, \email{Gregory.Kucherov@univ-eiffel.fr}
	\and
	Michigan Technological University, USA, \email{yakov@mtu.edu}
}

\maketitle

\begin{abstract}
	The normalized substring complexity $\delta$ of a string is defined as $\max_k \{c[k]/k\}$, where $c[k]$ is the number of \textit{distinct} substrings of length $k$. This simply defined measure has recently attracted attention due to its established relationship to popular string compression algorithms. 
	We consider the problem of computing $\delta$ online, when the string is provided from a stream. 
	We present two algorithms solving the problem: one working in $O(\log n)$ amortized time per character, and the other in $O(\log^3 n)$ worst-case time per character. To our knowledge, this is the first polylog-time online solution to this problem. 
\end{abstract}

\section{Introduction}
\label{sec:intro}
Compressibility of a string is a classical way to evaluate its \textit{complexity}. Compressibility has also other applications such as measuring similarity of two strings, where the similarity is expressed by how well one of the strings can be compressed if the other is provided as a reference \cite{DBLP:journals/tit/CilibrasiV05}. The ultimate compressibility measure is provided by Kolmogorov complexity which, however, is not practical in any way: it is not computable and is defined up to an additive constant and therefore is not meaningful for an individual finite string. A common practical approach is then to apply conventional string compressors based on efficient compression/decompression algorithms. 

Dictionary-based compressors is a class of compression algorithms trying to exploit repeated substrings occurring in the string. These include such popular families as Lempel-Ziv compressors, compressors based on the Burrows-Wheeler transform (BWT), grammar-based compressors, and more (see e.g. \cite{DBLP:journals/tit/KociumakaNP23}). 
A line of research 
focuses then on relationships between these algorithms with respect to the compressibility measures they provide \cite{10.1145/382780.382782,9317909,10.1145/3434399,10.1007/978-3-319-77404-6_36,DBLP:journals/tit/KociumakaNP23}. It is known, for example, that Lempel-Ziv (LZ77) compression of a string and the run-length encoding of its BWT image are related in resulting size by a $\mathrm{polylog}(n)$ factor ($n$ string length) \cite{9317909}. 

In this context, a natural measure of repetition-based compressibility is based on \textit{normalized substring complexity} defined as $\delta(w)=\max_{k\geq 1}\{c_w[k]/k\}$ where $c_w[k]$ is the number of \textit{distinct} subwords of length $k$ occurring in $w$. In relation to compression, $\delta$ was first studied in \cite{DBLP:journals/algorithmica/RaskhodnikovaRRS13} and later in \cite{10.1145/3426473,DBLP:journals/tit/KociumakaNP23,bernardini_et_al:LIPIcs.ISAAC.2023.12,DBLP:conf/dcc/BeckerCCKKP24}.  It has been shown that the compressibility measures provided by all the above-mentioned compression algorithms are bounded between $\delta$ and $\delta \cdot\mathrm{polylog}(n)$ \cite{DBLP:journals/tit/KociumakaNP23,9317909}. In practice, on genomic data, $\delta$ measure has been observed to be essentially the same as those based on LZ77 or BWT \cite{Bonnie2023.02.02.526837}. Furthermore, $\delta$ can be computed in $O(n)$ time and space by a simple and elegant algorithm based on suffix trees \cite{10.1145/3426473}. Finally, $\delta$ has a number of other attractive properties: it is invariant under string reversal or alphabet permutations, monotonically increasing under string concatenation, and sub-additive under the union of all substrings \cite{DBLP:conf/dcc/BeckerCCKKP24}. 

On the other hand, \textit{subword complexity}, i.e. counts of distinct subwords as a function of their length, is one of the established subjects in word combinatorics \cite{DBLP:journals/eatcs/BerstelK03}, and we will use some of these results in the present work. 

In this paper, we study the problem of computing $\delta$ in the \textit{online} mode. In light of the useful properties of $\delta$, it is important to be able to compute it online, especially for large input strings. Note that online implementation of compression algorithms is an extensively studied issue. In particular, a number of works has been done on efficient computation of Lempel-Ziv parsing, including online computation \cite{10.1007/978-3-540-87744-8_58,10.1007/978-3-642-32589-2_68,yamamoto_et_al:LIPIcs.STACS.2014.675,10.1007/978-3-319-23826-5_2}. Since $\delta$ is a good estimator of the Lempel-Ziv compression measure (cf above), maintaining its exact value of online is an interesting problem which, to our knowledge, has not been studied before. A related work includes \cite{DBLP:conf/dcc/BeckerCCKKP24} where a probabilistic sketch-based streaming algorithm for computing an \textit{approximation} of $\delta$ has been studied. Another related work is \cite{DBLP:journals/algorithmica/RaskhodnikovaRRS13} which focuses on a probabilistic \textit{sublinear} (non-online) algorithm for approximating $\delta$. 
In this paper, we propose two algorithms for computing the exact value of $\delta$ online: one working in amortized $O(\log n)$ time and another in worst-case $O(\log^3 n)$ time per character. 

\section{Preliminaries}
We index string positions from $1$, i.e. assume an input string $w=w[1]\ldots w[n]$ over an alphabet of size $\sigma$. \textit{Substring complexity} or \textit{substring count} $c_w[k]$, $k\leq n$ is defined to be the number of distinct substrings of length $k$ occurring in $w$. The \textit{normalized substring complexity} of $w$ is defined by $\delta(w)= \max_{1\leq k\leq n}\{c_w[k]/k\}$. We also denote $\widetilde{k}_w$ to be the largest $k$ for which the maximum is reached, i.e. the largest $k'$ s.t. $c_w[k']/k' = \max_k\{c_w[k]/k\}$. For brevity, we will simply write $\delta$ and $\widetilde{k}$ if $w$ is clear from the context. 

As a warm-up, we prove the following statement.
\begin{lemma}
	Let $\sigma\geq 2$. For every $w$, $n>3$, it holds $\delta\leq n/\log_\sigma n$. 
	\label{lemma1}
\end{lemma}
\begin{proof}
	Clearly, $c_w[k]\leq \min\{n,\sigma^k\}$ for all $k$. In particular, we have $\widetilde{k} \geq \log_\sigma c_w[\widetilde{k}]$. Then 
	$c_w[\widetilde{k}]/\widetilde{k} \leq c_w[\widetilde{k}]/\log_\sigma c_w[\widetilde{k}]\leq n/\log_\sigma n$ because $c_w[\widetilde{k}]\leq n$ and the function $x/\log_\sigma x$ is increasing 
	for $x\geq 3$. 
	The remaining cases are $c_w[\widetilde{k}]= 2$ and $c_w[\widetilde{k}]= 1$. $c_w[\widetilde{k}]= 2$ can only happen for $\sigma=2$ and $\widetilde{k}=1$, in which case $\delta=2\leq n/\log_2 n$ except for $n=3$. Finally, $c_w[\widetilde{k}]= 1$ can only happen for the unary alphabet. 
\end{proof}
Note that a slightly weaker inequality than Lemma~\ref{lemma1} follows from already known results relating $\delta$ to the number of phrases $z$ in LZ77 factorization. Namely, it is known that $\delta \leq z$ \cite{DBLP:journals/algorithmica/RaskhodnikovaRRS13} and $z\leq n/(\log_\sigma n-2-\log_\sigma(1+\log_\sigma n))$ shown in the original paper \cite{DBLP:journals/tit/LempelZ76}.

\section{Structure of the substring count array}
\label{sec:structure}
Given a string $w$ of length $n$, consider the array $c_w[1..n]$ of substring counts. 
It has a peculiar structure that we summarize in this section. 
To introduce these results, we need some definitions. 

For an occurrence of a substring $u=w[i..j]$, we call $w[i-1]$ its \textit{left context} and $w[j+1]$ its \textit{right context}. Furthermore, a substring $u$ of $w$ is called \textit{right-special} if $u$ occurs in $w$ with two distinct right contexts, i.e. $ux$ and $uy$ both occur in $w$ where $x$ and $y$ are two distinct letters. Similarly, a substring is \textit{left-special} if it occurs with two distinct left contexts. Finally, 
a substring is \textit{bispecial} if it is both \textit{right-special} and \textit{left-special}. It is known that substring complexity is closely related to special substrings \cite{Cassaigne_Nicolas_2010}. 

Consider a string $w$. Let $\alpha(w)$ be the length of the shortest suffix of $w$ that does not occur earlier in $w$. That is, substring $w[n-\alpha(w)+1..n]$ occurs only once in $w$. Let $\beta(w)$ be the smallest integer $\ell$ such that no length-$\ell$-substring of $w$  is right-special. In other words, $\beta(w)$ is one more than the maximal length of a right-special factor of $w$. Observe that $\alpha(w)>0$ and $\beta(w)>0$. 

\begin{example}
	For $w=\mathsf{abaabbabbab}$, $\alpha(w)=6$ and $\beta(w)=3$. 
	\label{ex1}
\end{example}

The structure of the array $c_w$ is summarized in the following theorem based on results from \cite{DBLP:journals/tcs/Luca99,leve2001proof}. 
\begin{theorem}[\cite{DBLP:journals/tcs/Luca99,leve2001proof}]
	Let $w$ be a string. The array $c_w[1..n]$ can be decomposed into three intervals $[1..L]$, $[L..R]$ and $[R..n]$ with the following properties:
	\begin{enumerate}[label=(\roman*)]
		\item counts $c_w[1..L]$ are strictly increasing, $c_w[k]<c_w[k+1]$ for $1\le k<L$\label{c1}
		\item counts $c_w[L..R]$ have equal values, $c_w[k]=c_w[k+1]$ for $L\le k\le R$ \label{c2}
		\item counts $c_w[R..n]$ are strictly decreasing so that $c_w[k+1]=c_w[k]-1$ and $c_w[n]=1$, \label{c3}
		\item if $\alpha(w)\geq \beta(w)$, then $L=\beta(w)$ and $R=\alpha(w)$, \label{alpha=R}
		\item  if $\alpha(w)\leq \beta(w)$, then $L\geq\alpha(w)$ and $R=\beta(w)$.\label{alpha<=L}
	\end{enumerate}
	\label{th1}
\end{theorem}
\begin{example}
	To continue Example~\ref{ex1}, for $w=\mathsf{abaabbabbab}$, the count array $c_w$ is $\mathsf{24666654321}$. Here $L=\beta(w)=3$ and $R=\alpha(w)=6$. If $\mathsf{a}$ is appended to $w$, then it becomes $\alpha(w\mathsf{a})=4$, $\beta(w\mathsf{a})=6$, and the array $c_{w\mathsf{a}}$ is $\mathsf{246777654321}$, so $L$ and $R$ are respectively $4$ and $6$ as well. Alternatively, if $\mathsf{b}$ is appended to $w$, then $\beta(w\mathsf{b})=3$, $\alpha(w\mathsf{b})=7$ and the array $c_{w\mathsf{c}}$ becomes $\mathsf{246666654321}$. Finally, if a new letter $\mathsf{c}$ is appended to $w$, then $\alpha(w\mathsf{c})=1$, $\beta(w\mathsf{c})=6$, and $c_{w\mathsf{c}}$ is $\mathsf{357777654321}$ with $L=3$ and $R=6$. 
	\label{ex2}
\end{example}

Theorem~\ref{th1} implies that the maximum of $c_w[k]/k$ can be reached only on interval $[1..L]$, that is $\widetilde{k}\in [1..L]$, as for $k>L$, we have $c_w[L]/L$>$c_w[k]/k$. 

\section{Updating the substring count array online}
\label{sec:updating}
Let us now focus on how the array $c_w$  can be modified when a letter $x$ is appended to $w$ on the right. Let $c_w[1..n]$ be the substring count array of $w$. We are interested in $c_{wx}[1..n+1]$. 
The following Lemma specifies the updates. 

\begin{lemma}
	Consider $\alpha(wx)$. For $k\in [1..\alpha(wx)-1]$, we have $c_{wx}[k]=c_{w}[k]$ . For $k\in [\alpha(wx)..n]$, we have $c_{wx}[k]=c_{w}[k]+1$, and $c_{wx}[n+1]=1$. 
	\label{lem1}
\end{lemma}
\begin{proof}Appending $x$ to $w$ introduces $n+1$ potentially new substrings which are suffixes of $wx$ of length $1$ to $n+1$. By definition, $\alpha(wx)$ is the length of the shortest suffix that has no other occurrences in $wx$. Therefore, the shorter suffixes of length $1$ to $\alpha(wx)-1$ are already present and the counts for $k=1..\alpha(wx)-1$ remain unchanged. For every $k\in [\alpha(wx)..n+1]$, exactly one new substring is  introduced. \end{proof}

Thus, appending a letter $x$ to $w$  implies incrementing by $1$ entries of $c_{w}$ starting from $\alpha(wx)$ to the end of the array. 

Assume now that input string $w[1..n]$ is provided online. We define $\alpha_i=\alpha(w[1..i])$. 
The following lemma clarifies how index $\alpha_{i+1}$ relates to $\alpha_i$ and to the intervals defined in Theorem~\ref*{th1}.

\begin{lemma}
	\begin{enumerate}[label=(\roman*)]
		\item $\alpha_{i+1}\leq \alpha_i+1$, \label{claim1}
		\item Let $L,R$ be  interval borders of $c_{w[1..i]}$ as defined in Theorem~\ref*{th1}. Then either $\alpha_{i+1}\leq L+1$ or $\alpha_{i+1}=R+1$ (or both if $L=R$). \label{claim2}
	\end{enumerate}
	\label{lem2}
\end{lemma}
\begin{proof} \ref{claim1} follows from the definition of $\alpha$, as the shortest non-repeated suffix of $w[1..i+1]$ cannot be more than by $1$ longer than that of $w[1..i]$. To show \ref{claim2}, first recall that by Theorem~\ref{th1}\ref{alpha<=L},\ref{alpha=R}, we have either $\alpha_i\leq L$ or $\alpha_i=R$ (or both, if $L=R$). If $\alpha_i\leq L$, then $\alpha_{i+1}\leq L+1$ follows from \ref{claim1}. Consider the case $\alpha_i = R$. 
	$L+1< \alpha_{i+1}\leq R$ would contradict the structure of the array $c_{w[1..i+1]}$ implied by Theorem~\ref*{th1}, namely that $c_{w[1..i+1]}$ must strictly increase on the first interval (Theorem~\ref*{th1}\ref{c1}). Therefore, we must have $\alpha_{i+1}=R+1$ or $\alpha_{i+1}\leq L+1$. \end{proof}
Observe that the case $\alpha_{i+1}= R+1$ extends the constant interval (\ref{c2} of Theorem~\ref*{th1}) by $1$ (i.e. $R$ gets incremented by $1$), the case $\alpha_{i+1}= L+1, L<R$ shortens the constant interval by $1$ ($L$ gets incremented by $1$), and the remaining cases keep $L$ and $R$ unchanged. The three cases are illustrated in Example~\ref{ex2}. 

From now on, we assume that string $w$ is provided from a stream, and we drop the index $w$ from $c_w$. We ``truncate'' array $c$ to contain only the increasing and constant intervals (see Theorem~\ref{th1}), that is, at each step $i$, we represent $c$ by $c[1..R]$ for the current value of $R$. It is easily seen that the descending ``tail'' $c[R+1..i]$ 
does not need to be stored to maintain $c$. 
From the above, it follows that at step $i$, $c[1..R]$ undergoes an update of one of two types: either all entries of $c[1..R]$ starting from position $\alpha_{i}$ are incremented by $1$ (see Lemma~\ref{lem1}), or $c$ is extended by setting $R \leftarrow R+1$ and $c[R]\leftarrow c[R-1]$. The latter corresponds to extending the constant interval, as explained above. 
To summarize, the algorithm for updating $c[1..R]$ after appending a next character to the current string $w[1..i-1]$ can be stated as follows:

\begin{algorithm}[H]
	\algo{\textsc{update}($c[1..R]$,$w[i]$)}{
		compute $\alpha_{i}=\alpha(w[1..i])$\;
		\eIf{$\alpha_{i}=R+1$}{$R\leftarrow R+1$; $c[R]\leftarrow c[R-1]$}
		{\For{$k\leftarrow \alpha_{i}$ \KwTo $R$}{$c[k]\leftarrow c[k]+1$}}
	}
\end{algorithm}

In the rest of this section, we focus on properties of $\alpha_i$'s over the execution of \textsc{update}. According to Lemma~\ref{lem2}(i), at each step $i$, either $\alpha_{i}=\alpha_{i-1}+1$, i.e. $\alpha_{i-1}$ is incremented by $1$, or $\alpha_i\leq \alpha_{i-1}$. We call the first case an \textit{increment} and the second a \textit{pullback}. The following lemma is immediate, it states that the sum of all pullback distances is bounded by $n$. 
\begin{lemma}
	$\sum_{i\,|\,\alpha_{i}\leq \alpha_{i-1}}(\alpha_{i-1}-\alpha_{i})\leq n$.
	\label{lem:rollbacks}
\end{lemma}
\begin{proof}
	As $\alpha_1=1$, there are no more than $n$ increments and therefore the sum of $(\alpha_{i-1}-\alpha_{i})$ over all pullbacks is bounded by $n$ as well.                        
\end{proof}
We also show that the sum of $\alpha_i$ resulting from the pullbacks is $O(n\log n)$.  						
\begin{theorem}
	$\sum_{i\,|\,\alpha_{i}\leq \alpha_{i-1}}\alpha_i=O(n\log n)$. 
	\label{th2}
\end{theorem}
The proof of Theorem~\ref{th2} is given in the Appendix. 

\section{Maintaining $\delta$ online}
We now turn to our main problem: maintaining normalized substring  complexity $\delta$ online. We follow the \texttt{update} algorithm which decomposes this problem into two: computing $\alpha_i$ and maintaining $\max\{c[k]/k\}$ on the array $c$ under updates specified by \textsc{update}. In this section, we describe efficient online solutions to each of the two problems.

\subsection{Computing $\alpha_i$}
\label{sec:alphai}
Computing $\alpha_i$ is a known problem in string algorithms: values $(\alpha_i-1)$ are lengths of \textit{longest repeated suffixes}, i.e. longest suffixes of $w[1..i]$ occurring earlier\footnote{Longest repeated suffixes are sometimes called \textit{longest previous factors} \cite{10.1007/978-3-540-87744-8_58}, however the latter term can also refer to longest earlier-occurring factors \textit{starting} at a given position, rather than ending, as e.g. in \cite{CROCHEMORE200875}.}. 
Perhaps the easiest way to compute $\alpha_i$'s online is to apply Ukkonen's online suffix tree construction algorithm \cite{DBLP:journals/algorithmica/Ukkonen95}. After reading $w[i]$, the algorithm computes $\alpha_{i}$ as a by-product. It maintains the node, called \textit{active point} in \cite{DBLP:journals/algorithmica/Ukkonen95}, corresponding to the longest suffix of $w[1..i]$ appearing earlier. The string length of this node is easily retrieved from the algorithm. $\alpha_{i}$ is one more than this length. Since Ukkonen's algorithm works in time $O(n\log \sigma)$, this gives a way to compute $\alpha_{i}$ in $O(\log \sigma)$ \textit{amortized} time per character, using $O(n)$ words of space. 

The above solution provides only an amortized time bound and can spend up to $O(n)$ time on an individual character. In~\cite{10.1007/978-3-540-87744-8_58} the authors present an algorithm that spends $O(\log^3 n)$ worst-case time on a character, and, moreover, uses a \textit{succinct} space of $n\log \sigma + o(n\log \sigma)$ bits where $\sigma$ is the alphabet size. In~\cite{10.1007/978-3-030-51466-2_31}  the time is improved to \textit{amortized} $O(\log^2 n)$, with even stronger condition of \textit{compressed} space. 

Here we briefly describe how the approach of \cite{10.1007/978-3-540-87744-8_58} can be modified to obtain $O(\log n)$ worst-case time per character, under a less stringent space requirement of $O(n)$ words. The algorithm of \cite{10.1007/978-3-540-87744-8_58} relies on the online construction of \textit{compressed prefix array} (CPA) for the streamed string. CPA is defined as the more common \textit{compressed suffix array} built on the inverted string, as it is constructed on a string provided right-to-left, as in Weiner's suffix tree algorithm. The CPA is represented by two data structures that we have to maintain. The first is the Burrows-Wheeler transform (BWT) of the current input string, on which we have to support \textsc{rank} and \textsc{select} queries under updates of inserting a character at any position of the BWT string. The second data structure is the \textit{longest common suffix} (LCS) array which is a dynamic array of numbers on which we have to support \textit{range minimum queries} (\textsc{rmq}) under updates of inserting a new entry and modifying an entry. An update to the whole CPA after appending a character to the text is implemented by a constant number of updates and queries to BWT string and LCS array. 

\textsc{rank} and \textsc{select} queries on a dynamic string is an extensively studied problem (see \cite{lee2009dynamic,doi:10.1137/130908245,MunroN15} and references therein). 
It can be solved in optimal $O(\log n/\log \log n)$ worst-case time on both updates and queries for strings on a general alphabet \cite{MunroN15}. 

The most time-consuming step of the algorithm of \cite{10.1007/978-3-540-87744-8_58} comes from \textsc{rmq} queries on a dynamic LCS array. Solutions proposed in the literature focus either on a compact representation of the LCS array itself \cite{10.1007/978-3-030-51466-2_31} or on saving the working space \cite{HELIOU2016108}, or provide amortized bounds \cite{brodal2011path}. However, if we allow $O(n)$ words of space, we can simply keep the array in a balanced range tree supporting both \textsc{rmq} queries and array updates in $O(\log n)$ time. 

We summarize this discussion in the following result. 
\begin{theorem}
	$\alpha_i$'s can be computed online in worst-case $O(\log n)$ time per character. 
\end{theorem}

\subsection{Convex hull reduction}
\label{sec:ch-reduction}

We now turn to the second problem: maintaining $\max\{c[k]/k\}$ under updates to $c$ array specified by \textsc{update}. We reduce the problem to a geometric problem of dynamically maintaining an (upper) convex hull on a certain set of points on the plane. 

We represent the sub-array $c[1..R]$  of the substring count array by a (dynamic) set of planar points $\Points=\{(k,c[k]), k=1..R\}$. 
Recall (see procedure \textsc{update}) that appending a letter to $w$ is equivalent  to modifying the set of points by one of the following two operations:
\begin{enumerate}[label=(\roman*)]
	\item given $\alpha\in [1..R]$, shift $(k,c[k])$ to $(k,c[k]+1)$ for each $k\in [\alpha..R]$,
	\item increment $R\leftarrow R+1$ and add a new point $(R,c[R])$ with $c[R]=c[R-1]$. 
\end{enumerate}
Operation (i) shifts all points $(k,c[k])$, $k\in [\alpha..R]$, upward by $1$, where value $\alpha$ is computed at the previous step, 
while operation (ii) adds a new point to the right of the current points.

The following observation guides our approach. 
\begin{lemma}
	Consider the point set $\Points=\{(k,c[k]), k=1..R\}$ and consider the upper convex hull of $\Points$. Let $(k',c[k'])$ be the point of tangency of the tangent line from the origin to the upper convex hull of $\Points$ (or one of those points, in case the tangent line coincides with an edge of the convex hull). Then $\delta(w)=c[k']/k'$. 
	\label{lem3}
\end{lemma}
\begin{proof}
	The ratio $c[k]/k$ is the slope of the line between the origin and $(k,c[k])$. Therefore, $\delta(w)$ corresponds to the point $(k',c[k'])\in\Points$ with the maximal slope. Clearly, $(k',c[k'])$ belongs to the upper convex hull and, on the other hand, to the tangent line from the origin to this convex hull. 
\end{proof}
\begin{figure}
	\centering
	\includegraphics[width=0.4\linewidth]{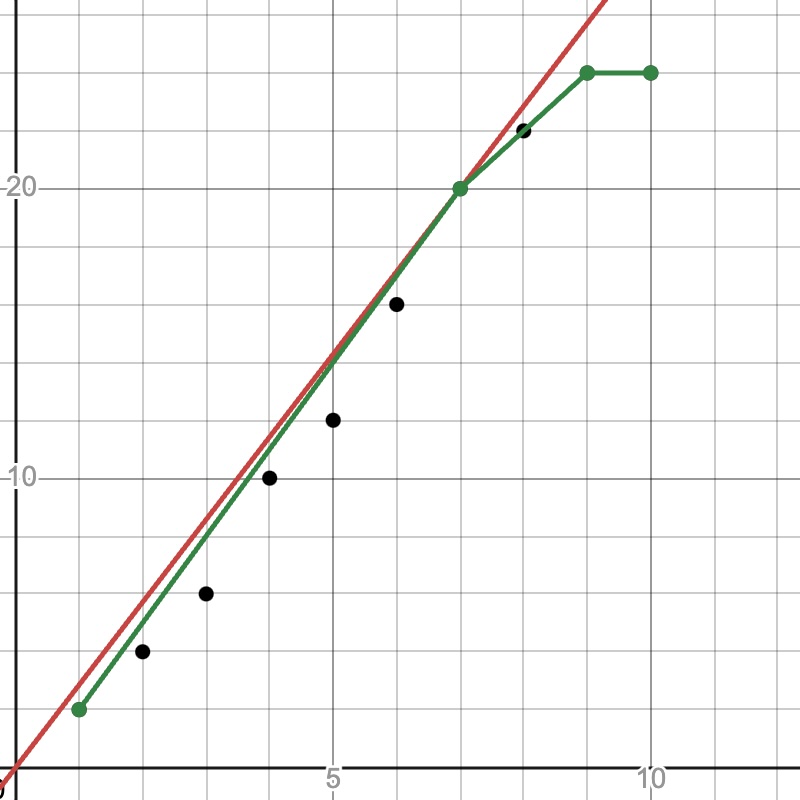}
	\caption{Point set $\Points$ for $w=011010011001011010010110011010011$, its upper convex hull (green) and tangent line (red). Here $R=10$. Tangency point $(7,20)$ corresponds to $\delta(w)=20/7$.
	}
	\label{fig:1}
\end{figure}
Figure~\ref{fig:1} illustrates Lemma~\ref{lem3}. 

The above observation will be used by our main algorithm presented in the next section. The algorithm will compute the tangent line to an interval of points and will use the result to maintain a point of maximum slope.

\subsection{Maintaining a point of maximum slope}
\label{sec:ouralg}

As follows from the above, the problem can be viewed as the one of maintaining the convex hull\footnote{Throughout this section, ``convex hull'' means ``upper convex hull''.} under operations (i) and (ii) from the previous section. 
Dynamic convex hull problem has been extensively studied in computational geometry, starting from the 70s \cite{DBLP:journals/cacm/Preparata79,OVERMARS1981166,10.5555/645413.757575,BrewerBW24} However, known data structures support updates of individual points (insertions/deletions) and don't support bulk updates required by our operation (i) and therefore cannot be applied to our problem in a blackbox fashion. 
In \cite{BrewerBW24}, it is shown that dynamic convex hull problem can be solved more efficiently when updates occur only at the extremities of the current set of points (along a fixed direction). 

In this section, we will demonstrate how the data structure of \cite{BrewerBW24} can be used to maintain the maximum slope under our updates in $O(n \log  n)$ time, that is in amortized $O(\log n)$ time per character. We summarize the results of \cite{BrewerBW24} that we need for our purposes. 
\begin{theorem}[\cite{BrewerBW24}]
	There exists a data structure supporting the following operations on a set of points $\mathcal{P} = \{(x,y)\}$. 
	\begin{itemize}
		\item inserting to $\mathcal{P}$ a point $(x',y')$ such that either $x'< \min_{(x,y)\in\mathcal{P}} x$ or $x'> \max_{(x,y)\in\mathcal{P}} x$, in worst-case time $O(1)$,
		\item deleting a point $(x',y')\in \mathcal{P}$ with either $x'=\min_{(x,y)\in\mathcal{P}} x$ or $x'= \max_{(x,y)\in\mathcal{P}} x$, in worst-case time $O(1)$,
		\item computing the tangent of the convex hull of $\mathcal{P}$ through a given point outside the convex hull and reporting the point of tangency, in time $O(\log n)$ where $n$ is the current number of points in $\mathcal{P}$. 
	\end{itemize}
	\label{thm:brewer}
\end{theorem}

The algorithm implied by \textsc{update} (Sect.~\ref{sec:updating}) is a sequence of $n$ steps, so that at each step $i$ we (i) either increment the values in some contiguous right part $c[\alpha_i..R_i]$\footnote{In this section, we use notation $R_i$ to refer to the value of $R$ at step $i$} of the current array $c$, or (ii) extend array $c$ on the right by duplicating the last entry. Along this process, we have to maintain $\max_{1\leq k\leq R_i}\{c[k]/k\}$. 
We call the \emph{active interval} the fragment $c[\alpha_i..R_i]$ that is incremented at the current step $i$. Note that for operation (ii), the active interval is $c[R_i..R_i]$, as this case occurs when $\alpha_i=R_{i-1}+1$ (see Section~\ref{sec:updating}). 

For an interval $\mathcal{I}\subseteq [1..R_i]$, we say that $\widetilde{k}_\mathcal{I}=\argmax_{k\in \mathcal{I}}\{c[k]/k\}$ is the \textit{leader on the interval} $\mathcal{I}$, or simply the \textit{leader} if $\mathcal{I}=[1..R_i]$. 

By Lemma~\ref{lem2}(i), $\alpha_{i}\leq \alpha_{i-1}+1$, therefore at each step, the left border $\alpha_i$ of the active interval is either incremented by $1$, or decremented or stays unchanged. 
We will decompose the algorithm into \emph{rounds} that are maximal sequences of steps with $\alpha_i$ strictly increasing, i.e. steps $i$ with $\alpha_{i}=\alpha_{i-1}+1$. Each round is followed by a pullback, i.e. a step $i$ 
such that $\alpha_{i}\leq \alpha_{i-1}$. Note that by Lemma~\ref{lem:rollbacks}, the sum of $(\alpha_{i-1} - \alpha_{i})$ over all pullback steps is bounded by $n$. 

We dynamically maintain the active interval $c[\alpha_i .. R_i]$ and the interval $c[1..\alpha_i-1]$ before the active one, in two different ways. 
For the interval $c[1..\alpha_i-1]$, we will explicitly store values $c[k]$  and maintain the leader $\argmax_{1\le k\le \alpha_i-1}\frac{c[k]}{k}$. To this aim, we use two arrays $\mathit{val}$ and $\mathit{leader}$ where
$\mathit{val}[k]=c[k]$ and $\mathit{leader}[k]=\argmax_{1\le j \le k}\frac{val[j]}{j}$. Arrays $\mathit{val}$ and $\mathit{leader}$ are dynamic and may grow in time, we implement them using the standard doubling technique.
Since within each round, we have $\alpha_{i}= \alpha_{i-1}+1$, updating $\mathit{val}$ and $\mathit{leader}$ takes $O(1)$ operations at each step. For the pullback step, no updates of $\mathit{val}$ and $\mathit{leader}$ are needed. 

For the active interval $c[\alpha_i .. R_i]$, we maintain the leader using the reduction to the convex hull problem explained earlier. We will maintain the convex hull of the active interval using the data structure of Theorem~\ref{thm:brewer}. However, each point $(k,c[k])$ of the active interval will be represented by point $(k,c[k]-\Delta)$, for a dynamic parameter $\Delta$ called the \emph{active shift}. The meaning of $\Delta$ is the vertical shift of the actual point $(k,c[k])$ relative to the $y$-coordinate of the point currently stored in the data structure. Our algorithm will ensure that the value of  $\Delta$ is the same for all points in the active interval of the array. 

We now explain how a step of our algorithm is implemented.
We ensure that the following invariants are satisfied after each step: 
\begin{invariant}
	\label{inv1}
	For any $k$ in the active interval $[\alpha_i..R_i]$, the convex hull data structure contains exactly one point 
	$(k,c[k]-\Delta)$ with the correct value of $\Delta$. 
\end{invariant}
\begin{invariant}
	\label{inv2}
	For any $k$ in the interval $[1..\alpha_i-1]$, all values $\mathit{val}[k]$ and $\mathit{leader}[k]$ are correct.    
\end{invariant}
At initialization, after the first character of the text is read, we have $c[1]=1$, $\alpha_1=R_1=1$, and we insert point $(1,c[1])$ into the convex hull structure, set $\Delta=0$, and initialize empty arrays $\mathit{val}$ and $\mathit{leader}$. 

At each step $i\ge 2$, we proceed as follows. First observe that operation (ii) is easily implemented by retrieving the rightmost point $(R_{i-1},c[R_{i-1}]-\Delta)$ from the convex hull data structure\footnote{The convex hull data structure described in \cite[Theorem 5]{BrewerBW24} supports the retrieval of extreme points in $O(1)$ time. Alternatively, we can maintain a separate list of all points of the active interval sorted by $x$-coordinates.} and inserting a new rightmost point $(R_{i},c[R_{i-1}]-\Delta)$ for $R_{i}= R_{i-1}+1$. Consider now operation (i). If  $\alpha_{i}=\alpha_{i-1}+1$, we retrieve the current leftmost point  $(\alpha_{i-1},c[\alpha_{i-1}]-\Delta)$ from the convex hull data structure and determine the value $c[\alpha_{i-1}]$ by adding $\Delta$, then update $\mathit{val}$ and $\mathit{leader}$. Specifically, we set $\mathit{val}[\alpha_{i-1}]\leftarrow c[\alpha_{i-1}]$ and set $\mathit{leader}[\alpha_{i-1}]\leftarrow \alpha_{i-1}$ if $\mathit{val}[\alpha_{i-1}]/\alpha_{i-1}\geq \mathit{val}[\mathit{leader}[\alpha_{i-1}-1]]/\mathit{leader}[\alpha_{i-1}-1]$ and $\mathit{leader}[\alpha_{i-1}]\leftarrow \mathit{leader}[\alpha_{i-1}-1]$ otherwise. Then we remove the leftmost point $(\alpha_{i-1},c[\alpha_{i-1}]-\Delta)$ from the convex hull data structure and increment $\Delta$ by $1$. Otherwise, if $\alpha_{i}\le \alpha_{i-1}$, we execute the pullback step. In this case, we consecutively insert points $(k,\mathit{val}[k]-\Delta)$ for all $k$ from $\alpha_{i-1}-1$ down to $\alpha_{i}$ into the convex hull data structure. We observe that every inserted point is to the left of all points already stored in the data structure. Then we increment $\Delta$ by $1$.  

\begin{lemma}
	\label{lemma:invar}
	The above algorithm correctly maintains Invariant~\ref{inv1}. In other words, for each $k\in [\alpha_i .. R_i]$ and any point $(k,\nu)$ stored in the convex hull data structure, $c[k]=\nu + \Delta$, for the current value of $\Delta$. 
\end{lemma}
\begin{proof}
	The Lemma follows by induction on steps. Invariant~\ref{inv1} is verified at initialization. Suppose it is satisfied after $(i-1)$ steps for some $i>1$.  If step $i$ is an increment step, the claim is straightforward. If step $i$ is a pullback step, we insert points $(k,c[k]-\Delta)$ for all $k$ such that $\alpha_i\le k < \alpha_{i-1}$, which guarantees that the active shift remains the same for all points in the data structure. 
\end{proof}

\begin{lemma}
	\label{lemma:invar2}
	The above algorithm correctly maintains  Invariant~\ref{inv2}. 
\end{lemma}
\begin{proof}
	Directly follows from the description of the algorithm. 
\end{proof}

At each step, we can compute the leader by taking the maximum of leader values on the two intervals.  
Given a set $P$ of planar points and $d>0$, 
we will say that the set $P'$ is a \textit{vertical $d$-shift} of $P$ iff for every point $p=(x,y)\in P$ there is exactly one point $p'=(x,y-d)\in P'$. 
To compute the leader on the active interval, we answer a tangent query to the convex hull, based on the following straightforward observation.
\begin{lemma}
	\label{lemma:shift1}
	Let $P$ be a set of planar points, $q=(\overline{x},\overline{y})$ a point outside the convex hull of $P$, and $P'$ the vertical $d$-shift of $P$ for some $d$. Let $H$ and $H'$ be the convex hulls of $P$ and $P'$ respectively. Then the tangent from $q$ to $H$ passes through point $p=(x,y)\in H$ iff the tangent from $q'=(\overline{x},\overline{y}-d)$ to $H'$ passes through point $p'=(x,y-d)\in H'$.
\end{lemma}

Let $\Points_a=\{(k,c[k]), k=\alpha_i..R\}$ be the set of points of the current active interval. By Lemma~\ref{lemma:invar}, the set of points $\Points'_a$, stored in our convex hull structure,  is a vertical $\Delta$-shift of $\Points$. 
Let $p'_a$ be the tangent point of $(0,-\Delta)$ and the convex hull of $\Points'_a$.  Let $p_a$ be the tangency point  from $(0,0)$ to the convex  hull of $\Points_a$.  By Lemma~\ref{lemma:shift1} $p$ and $p_a$ have  the same $x$-coordinates and their $y$-coordinates differ by $\Delta$: if $p'_a=(\widetilde{k}_a,\nu)$, then $p_a=(\widetilde{k}_a,\nu+\Delta)$. 
Hence $\widetilde{k}_a$ is the leader in the active interval and its value is $\nu+\Delta$. By Theorem~\ref{thm:brewer}, we can find $p'_a$ in $O(\log n)$ time using our convex hull data structure. The resulting leader after step $i$ is either $\widetilde{k}_a$ or $\mathit{leader}[\alpha_i-1]$, depending on which of the fractions $(\nu+\Delta)/\widetilde{k}_a$ or $\mathit{val}[\mathit{leader}[\alpha_i-1]]/\mathit{leader}[\alpha_i-1]$ is larger.

The complexity of the algorithm follows from Theorem~\ref{thm:brewer}. Step (ii) involves a retrieval and an insertion and therefore is done in $O(1)$ time. An incremental step (i) involves a retrieval and a deletion and again, takes $O(1)$ time. A pullback involves an insertion of each of the $\alpha_i-\alpha_{i+1}$ points. By Lemma~\ref{lem:rollbacks}, these take $O(n)$ time altogether. We also have to answer $O(n)$ tangent queries which takes $O(n \log n)$ time. All updates of arrays $\mathit{val}$ and $\mathit{maxval}$ take $O(n)$ time. Thus, the algorithm runs in time $O(n\log n)$. 

\begin{theorem}
	\label{theor:runtime}
	An online algorithm for computing the values of $\delta$ can be implemented in $O(n \log  n)$ time.
\end{theorem}

\subsection{Worst-case solution}
\label{sec:worst-case}

Due to pullback steps, Theorem~\ref{theor:runtime} provides only an amortized $O(\log n)$ bound on processing an individual character. In this section, we are interested in worst-case bounds, and we show how a bound of $O(\log^3 n)$ per character can be obtained. 

Here we rely on the technique of \cite{OVERMARS1981166} for maintaining a convex hull in worst-case $O(\log^2 n)$ time on each insertion/deletion of a point. Their method maintains the points of the upper convex hull\footnote{Note that \cite{OVERMARS1981166} considers left or right convex hulls with points ordered by $y$-coordinate, whereas here we consider upper convex hulls with points ordered by $x$-coordinate, which obviously does not affect the results.} sorted by $x$-coordinates, therefore tangent queries are supported in $O(\log n)$ time \cite{DBLP:journals/cacm/Preparata79}. 

We first give an overview of the data structure  of \cite{OVERMARS1981166} and then explain how their method can be adapted  to support our updates (i) and (ii) (Section~\ref{sec:ch-reduction}). The algorithm of \cite{OVERMARS1981166} maintains a balanced binary tree $\Tree$ where leaves hold points of the current set sorted by $x$-coordinate, and each internal node $\gamma$ of $\Tree$ is associated with an interval 
$\mathcal{I}_\gamma$ of points along the $x$-axis. 
A node $\gamma$ holds the upper convex hull $\Hull_{\gamma}$ of point set $\{(k,c[k]), k\in \mathcal{I}_\gamma\}$.   
Each $\Hull_{\gamma}$ is a list of points sorted by $x$-coordinate and implemented in a \textit{concatenable queue} data structure supporting splitting the list at a given point or concatenating two lists, each in $O(\log n)$ time. These operations can be implemented on  a search tree, such as a 2-3 tree (see e.g. \cite{10.5555/578775}). 

In \cite{OVERMARS1981166}, it is shown that the convex hull of a parent node can be computed from convex hulls of its child nodes in $O(\log n)$ operations. The sequence of points forming the convex hull of a parent node is obtained by concatenating some prefix of the sequence of the left child with some suffix of the sequence of the right child. As shown in \cite{OVERMARS1981166}, computing the ``break points'' in both sequences is done using $O(\log n)$ operations. 

Storing all convex hulls in all internal nodes, however, would be too expensive. A key idea of \cite{OVERMARS1981166} is to store, at each internal node, only the part of the sequence that does not contribute to the parent sequence. This way, each point belonging to some convex hull is stored only once, insuring a linear space for storing all convex hulls. Since splitting and concatenating is done in $O(\log n)$ operations, each convex hull can be reconstructed by traversing a path to this node, taking $O(\log^2 n)$ operations in total. Inversely, if e.g. the $y$-coordinate of a point is modified or a new point is inserted, the convex hulls along the corresponding path should be updated, which can be done in $O(\log^2 n)$ operations as well. 

Now we explain how we modify the data structure of ~\cite{OVERMARS1981166} to support our operations (i) and (ii). 
The main 
difference is that we do not store in $\Tree$ the $y$-coordinates of points: every point $(k,c[k])$ is represented in $\Tree$ by its $x$-coordinate $k$ only.  
$y$-coordinates $c[k]$ are stored in a 
separate dynamic  range tree $\Tree_c$. For every point $(k,c[k])$, there is a leaf of $\Tree_c$ that corresponds to it. Similar to $\Tree$, leaves of $\Tree_c$ are sorted by $x$-coordinates of points. Each time we need to retrieve the $y$-coordinate of some point $(k,c[k])$ we use $\Tree_c$ for this purpose. 

Operation (i) is implemented on $\Tree_c$ in the standard way by recording the increment in the root nodes of up to $\log n$ subtrees covering the interval $[\alpha_i..R]$. Values $c[k]$ are retrieved in time $O(\log n)$ by traversing the path to the leaf  that holds $k$ (lazy propagation). Operation (ii) (appending a new point) is supported in $O(\log n)$ time as well. 

To update $\Tree$ under operation (i), the key observation is that for nodes $\gamma$ such that $\mathcal{I}_\gamma\subseteq [1..\alpha_i -1]$ or $\mathcal{I}_\gamma\subseteq [\alpha_i..R]$, the convex hull does not change and therefore $\Hull_\gamma$ does not need to be modified. This is because operation (i) shifts all points with $x$-coordinates in $[\alpha_i..R]$ upward by $1$, which does not affect $\Hull_{\gamma}$ if $I_{\gamma}$ is contained in $[\alpha_i..R]$ or if $I_{\gamma}$ does not overlap with $[\alpha_i..R]$. 
Thus, the only convex hulls that should be updated are those for nodes $\gamma$ with $I_{\gamma}$ intersecting both $[1..\alpha_i-1]$ and $[\alpha_i..R]$. There are at most $\log n$ such nodes and they all belong to a single path in the tree, therefore these updates can be done in $O(\log^2 n)$ operations. Recall however that we have an additional $\log n$ factor coming from the fact that values $c[k]$ ($y$-coordinates of points) are retrieved in $O(\log n)$ time. Therefore, the overall time to support operation (i) is $O(\log^3 n)$. 

In case of operation (ii), we add a new leaf to $\Tree$ holding a new point $(R,c[R])$. Then we visit the ancestor nodes of the new leaf and update the convex hulls $H_{\gamma}$ in the same way as  described for operation (i). 

We conclude the discussion with the following
\begin{lemma}
	Operations (i) and (ii) can be supported in $O(\log^3 n)$ worst-case time. 
	\label{lemma:op-worstcase}
\end{lemma}

Tangent queries can be supported in $O(\log n)$ operations by a binary search algorithm \cite{DBLP:journals/cacm/Preparata79} and we have an additional $O(\log n)$ factor coming from the retrieval of $y$-coordinates from $\Tree_c$, resulting in $O(\log^2 n)$ time for querying the data structure for the tangency point. Taking into account that $\alpha_i$'s can be maintained in worst-case {$O(\log n)$ time} (see Section~\ref{sec:alphai}), Lemma~\ref{lemma:op-worstcase} implies the final result.

\begin{theorem}
	Normalized string complexity $\delta$ can be maintained online in $O(\log^3 n)$ worst-case time per character. 
	\label{thm:worst-case}
\end{theorem}

\section{Concluding remarks}

We showed how the normalized string complexity $\delta$ can be maintained online in $O(\log n)$ amoritized time and $O(\log^3 n)$ worst-case time per character. Note that for the first bound, the ``bottleneck'' is in supporting tangent queries, as the method of \cite{BrewerBW24} supports updates in constant time (see Theorem~\ref{thm:brewer}). This suggests that the $O(\log n)$ bound might not be optimal. 
On the other hand, for the $O(\log^3 n)$ worst-case bound, the bottleneck is in supporting the updates in worst-case time, and in particular, supporting a constant upshift on an interval of points. {We consider that efficiently maintaining the convex hull under this type of updates deserves to be further studied as well.} 

It is interesting to relate our results to those of \cite{DBLP:conf/waoa/Whittington24} about online computation of the smallest \textit{attractor} of a string. An attractor is a set of positions of the string such that at least one occurrence of every \textit{distinct} substring is hit by one of those positions \cite{10.1145/3188745.3188814}. It has been shown that the size $\gamma$ of the smallest attractor is another useful compressibility measure and has a deep relationship to many dictionary-based compressors (cf Introduction) \cite{10.1145/3188745.3188814,10.1145/3426473,DBLP:journals/tit/KociumakaNP23}. In particular, it is easy to show that $\delta\leq \gamma$ and, on the other hand, there are string families such that $\gamma = \Omega(\delta\log(n/\delta))$ \cite{DBLP:journals/tit/KociumakaNP23}. In contrast to $\delta$ which can be computed in linear time, computing $\gamma$ is an NP-complete problem \cite{10.1145/3188745.3188814}. In \cite{DBLP:conf/waoa/Whittington24}, it was shown that \textit{any} online algorithm has a competitive ratio $\Omega(\log n)$, i.e. can compute $\gamma$ only within an $\Omega(\log n)$ factor, in the worst case. On the other hand, this bound is matched by a greedy algorithm. In the context of these results, our work shows that, in accordance with the offline case, $\delta$ can be efficiently computed in the online mode as well.

\bibliography{biblio.bib}

\section*{Appendix: Proof of Theorem~\ref{th2}}
\renewcommand{\thetheorem}{\ref{th2}}
\begin{theorem}
	$\sum_{i\,|\,\alpha_{i}\leq \alpha_{i-1}}\alpha_i=O(n\log n)$. 
\end{theorem}
\begin{proof}
	Consider a pullback step $i$. At the previous step $i-1$, $w[i-\alpha_{i-1}..i-1]$ is the shortest suffix that does not occur earlier, and $w[i-\alpha_{i-1}+1..i-1]$ is the longest suffix that does occur earlier. Note that that its any earlier occurrence must be preceded by another letter than $w[i-\alpha_{i-1}]$, i.e. $w[i-\alpha_{i-1}+1..i-1]$ is left-special and occurs with a new left context\footnote{The earlier occurrence can also be a prefix, therefore we assume that any prefix is preceded by a special letter that does not occur in the rest of $w$. In other words, we assume that $w$ starts with a special unique letter.}. When letter $w[i]$ is processed, a pullback occurs when $\alpha_i\leq \alpha_{i-1}$ holds, that is, when the shortest unique suffix is strictly shorter than $\alpha_{i-1}+1$. This implies that $w[i-\alpha_{i-1}+1..i]$ does not occur earlier, that is any earlier occurrence of $w[i-\alpha_{i-1}+1..i-1]$ is followed by another letter than $w[i]$. Therefore, $w[i-\alpha_{i-1}+1..i-1]$ is right-special and occurs with a new right context. We conclude that $w[i-\alpha_{i-1}+1..i-1]$ is repeated and has both left and right contexts that are new. 
	
	In~\cite{DBLP:conf/cpm/KarkkainenMP09} the authors study \textit{irreducible LCP} (longest common prefix) values and prove that their sum is bounded by $2n\log n$. An LCP value is the length of the longest common prefix between a suffix and the lexicographically preceding one. 
	An LCP value is irreducible if these two suffixes are preceded by distinct letters. This implies that if step $i$ is a pullback, then it generates a new irreducible LCP value $(\alpha_{i-1}-1)$ which is the LCP between the suffix $w[i-\alpha_{i-1}..i]$ and a lexicographically adjacent suffix starting with one of the preceding copies of $w[i-\alpha_{i-1}..i-1]$. 
	
	Finally we need to show that when the string will be further extended, there always be a distinct LCP value for every $\alpha_{i}$ resulting from a pullback. 
	Let $v=w[i-\alpha_{i-1}..i-1]$. After the pullback, the new irreducible LCP value is formed by two lexicographically adjacent suffixes starting with $vx_1$ and $vx_2$ for letters $x_1\neq x_2$ and preceded by letters $y_1$ and $y_2$, $y_1\neq y_2$, respectively. When the string is extended, it can happen that a new suffix gets inserted between these two suffixes, that might potentially cancel or modify this irreducible LCP value so that $|v|=(\alpha_{i-1}-1)$ would not be properly accounted for in the final sum of irreducible LCP values. We show that this is not the case. 
	
	Assume that a new suffix gets inserted between the two suffixes starting with $vx_1$ and $vx_2$. Obviously, this suffix starts with $v$ and is followed by a letter which is either $x_1$ or $x_2$, or a letter that is between them in the lexicographical order. Consider the letter $z$ preceding this suffix in the string. If $z$ is different from both $y_1$ and $y_2$, then two new irreducible LCP values replace the former one. One of them may correspond to a new pullback, but the other one is necessarily larger than or equal to $|v|$. If $z=y_1$ or $z=y_2$, then the new middle suffix cannot result from a pullback. On the other hand, the three suffixes form one irreducible LCP value which is again larger than or equal to $|v|$. By induction, we conclude that each $\alpha_i$ resulting from a pullback is properly charged on a distinct LCP value in the final sum of irreducible LCP values. 
	The theorem follows. 
\end{proof}

Note that the sum of irreducible LCP values can be bounded more precisely in terms of compression measures: it is proven in \cite{9317909} that this sum is $O(n\log\delta)$. 
Thus, the proof of Theorem~\ref{th2} implies $\sum_{i\,|\,\alpha_{i}\leq \alpha_{i-1}}\alpha_i=O(n\log \delta)$.
\end{document}